\documentclass[aps,pra,twocolumn,showkeys]{revtex4-1}
\usepackage{amsmath,paralist,amsthm,comment,amssymb}
\usepackage{times,color}
\usepackage{epsfig}

\usepackage{tikz}
\usetikzlibrary{matrix,arrows}
\usetikzlibrary{shapes.geometric}
\usetikzlibrary{patterns} 
\usepackage{pgfplots}
\input{Qcircuit}

\DeclareMathOperator{\supp}{supp}

\newcommand{\mc}[1]{\mathcal{#1}}

\newcommand{\ms}[1]{\mathsf{#1}}
\newcommand{\C}{\mathbb{C}}
\newcommand{\F}{\mathbb{F}}
\newcommand{\nix}[1]{}

\newcommand{\be}{\begin{eqnarray*}}
\newcommand{\ee}{\end{eqnarray*}}
\newcommand{\ben}{\begin{eqnarray}}
\newcommand{\een}{\end{eqnarray}}

\newcommand{\ba}{\begin{array}}
\newcommand{\ea}{\end{array}}
\newcommand{\bmt}{\left[\begin{array}}
\newcommand{\emt}{\end{array}\right]}

\newtheorem{theorem}{Theorem}
\newtheorem{corollary}[theorem]{Corollary}
\newtheorem{lemma}[theorem]{Lemma}
\newtheorem{proposition}{Proposition}

\newtheorem{remark}{Remark}

\begin{document}

\title{Non-Threshold Quantum Secret Sharing Schemes in the Graph State Formalism}

\author{Pradeep Sarvepalli}
\email[]{pradeep.sarvepalli@gatech.edu}
\affiliation{Department of Chemistry and Biochemistry, Georgia Institute of Technology, Atlanta, GA 30332}
\date{February 15, 2012}

\begin{abstract}
In a recent work, Markham and Sanders have proposed a framework to study quantum  secret sharing  (QSS)  schemes using graph states.  This framework  unified  three classes of QSS protocols, namely, sharing classical secrets over private and public channels,  and sharing quantum secrets. However,  most  work on  secret sharing based on graph states  focused on threshold schemes. In this paper, we focus on general access structures. We show how to realize a large  class of arbitrary access structures using the graph state formalism. We show an  equivalence between  $[[n,1]]$ binary quantum codes and graph state secret sharing schemes sharing one bit. We also establish a similar (but restricted) equivalence between a class of $[[n,1]]$  Calderbank-Shor-Steane (CSS) codes and graph state QSS schemes sharing one qubit. With these results we are able to construct a large class of quantum secret sharing schemes with arbitrary access structures.
\end{abstract}


\maketitle
\section{Introduction}

Quantum secret sharing (QSS) \cite{hillery99, cleve99} deals with the problem of sharing classical or 
quantum secrets using quantum information. Further, the secret sharing protocol  could be operated in
the presence or absence of eavesdroppers.  In \cite{markham08}, a graph state formalism was proposed 
with a view to unifying all these variants under the same umbrella. This framework was useful in ways
other than unifying the  various quantum secret sharing protocols. For instance, building upon this 
framework, researchers have  been able to  propose new secret sharing protocols  \cite{javelle11},  and make a connection with the 
measurment-based quantum computation model \cite{kashefi09}. 
More recently, it has motivated research in graph theoretic concepts such as weak odd domination 
 \cite{gravier11}.

The graph state framework  does  in principle include non-threshold access
structures for quantum secrets. However, neither \cite{markham08} nor subsequent works 
\cite{keet10,kashefi09,javelle11,gravier11}  provide any procedure to explicitly construct schemes
with arbitrary access structures in the graph state formalism. The graph state framework in 
for quantum secret sharing approaches it from a perspective other than quantum error-correction, in contrast to the theory as developed in \cite{cleve99,gottesman00,smith00}. 
But since any secret sharing protocol is ultimately an error-correcting code, the graph state schemes must be equivalent to those based on
stabilizer codes and the protocols in \cite{markham08} must arise from quantum codes.
But no results are known in this direction. 

The main contribution of this paper is to fill these gaps. We make transparent the connection between the graph state framework and the protocols presented using quantum codes. 
We show an equivalence between $[[n,1]]$ binary quantum codes  and graph state protocols sharing one bit. We show a restricted equivalence between a class of $[[n,1]]$ CSS codes \cite{calderbank98} and protocols sharing one qubit. 
We also  translate  many of the schemes developed using quantum codes into 
those based on the graph state formalism.

We emphasize that our results are constructive and provide concrete details for the construction of the secret sharing
schemes as well as the associated details of recovery. 

We restrict ourselves to the qubit case in this paper, although as shown in \cite{keet10} graph state
secret sharing schemes   can be extended to other alphabet.
The general case involving qudits will be explored elsewhere.

\section{Background}

\subsection{Quantum secret sharing}
We briefly review the pertinent ideas of quantum secret sharing. We assume that the
reader is familiar with  quantum codes and the stabilizer formalism \cite{calderbank98,gottesman97}. In 
a secret sharing scheme,  a dealer distributes an encrypted secret to a collection of players. Then 
certain subsets of players can collaboratively  reconstruct the secret. Those subsets which can recover 
the secret are called authorized sets and those
that cannot are said to be unauthorized sets.  The collection of authorized sets is called the 
access structure of the scheme, which we denote as $\Gamma$. For an access structure to be valid, it must be monotonic, i.e., 
any set that contains an authorized set must also be an authorized set. 
An authorized set is said to be minimal if any proper subset of it is unauthorized. The collection of minimal authorized sets is called the minimal access structure. 
In a  threshold scheme with threshold $k$, any subset consisting of $k$ or more players can access the secret
while those with fewer players cannot. In a general access structure, the authorized sets can be of
different sizes and all subsets of that size need not be authorized. 
A collection of sets $\Gamma_{\rm{gen}}$ is said to generate the access structure $\Gamma$, if every 
authorized set contains some element of $\Gamma_{\rm{gen}}$. 

A secret sharing scheme is
said to be perfect if the unauthorized sets cannot extract any information about the secret. In this 
paper we only consider perfect secret sharing schemes.

When the secret to be shared is classical, the dealer distributes  a set of
orthogonal quantum states that encode the secret. 
The following result, due to Gottesman, states the conditions that must be satisfied by authorized and
unauthorized sets for sharing classical secrets through a QSS scheme. 

\begin{proposition}[Access conditions for classical secrets, \cite{gottesman00}]\label{lm:accessConditions}
Suppose we have a set of orthonormal states $\ket{\psi_i}$ encoding a classical secret. Then a set $T$
is an unauthorized set iff
\begin{eqnarray}
\bra {\psi_i} F \ket{\psi_i} =c(F)\label{eq:unauthSet}
\end{eqnarray}
independent of i for all operators $F$ on $T$. The set $T$ is authorized iff
\begin{eqnarray}
\bra {\psi_i} E \ket{\psi_j} =0 \quad (i\neq j)\label{eq:authSet}
\end{eqnarray}
for all operators $E$ on the complement of $T$.
\end{proposition}

If we were to share a quantum secret, then the access structure, in addition to being monotonic, must
also satisfy the no-cloning theorem \cite{cleve99}. This implies that no two authorized sets
are disjoint. In this case the access structure must satisfy the conditions of 
Proposition~\ref{lm:accessConditions} for any state in the space spanned by the encoded states, see
\cite[Theorem~1]{gottesman00}.

\subsection{Review of graph state formalism for quantum secret sharing}

In \cite{markham08}, the quantum secret sharing protocols were classified as follows:
i) CC--This protocol deals with the sharing of classical secrets, where we assume that the players have secure channels.
ii) CQ--In this protocol we share classical secrets where we assume that the channels between the players are susceptible to eavesdropping.
iii) QQ--This protocol shares quantum secrets using quantum channels. 
In this paper we restrict our attention to CC and QQ protocols. 

Let $\mathsf{G}$ be a graph with vertex set $V(\mathsf{G})$. We denote the neighbours of a vertex $v \in V(G)$  as $N_v$. We denote the graph
obtained by deleting the vertex $v$ from $\mathsf{G}$, by $\mathsf{G}\setminus v$. 
The graph state defined on $\mathsf{G}$ is  denoted  $\ket{\mathsf{G}}$.  Recall that the graph state
is a stabilizer state and satisfies $K_v\ket{\ms{G}}=\ket{\ms{G}}$, where 
\ben
K_v = X_v\prod_{u\in N_v} Z_u, \mbox{ for all }v \in V(\ms{G}).\label{eq:stabGen}
\een
We use the notation $K_A=\prod_{i\in A} K_i$. The stabilizer of $\ket{\ms{G}}$ is denoted as 
$S(\ket{\ms{G}})$.

In the  CC quantum secret sharing protocol, the secret bit $s$  
is encoded as 
\ben 
\mc{E} : s\mapsto Z_A^s \ket{\mathsf{G}},
\een
where $Z_A^s = \prod_{i\in A} Z_i^s$.
We denote a CC protocol using the graph $\ms{G}$ and encoding using the set $A$ by 
$(\ms{G}, A)$.
An authorized set $T$ can recover the secret by either performing a 
joint measurement of an appropriate operator  $M\in S(\ket{\ms{G}})$ or by local
measurements and combine these results classically (after
classical communication), in other words through LOCC.

In the QQ protocol, the dealer needs to add an additional ancilla qubit whose
state is the secret to be shared. The dealer then encodes this  state by a procedure similar to 
teleporation. Following this the dealer might have to perform some correction
operations on the encoded state to ensure that the secret has been properly
teleported. The dealer then distributes the qubits to the players. 
In this setting, authorized subsets of players can reconstruct the secret by 
means of suitable  nonlocal operations.

In \cite{kashefi09}, the graph state secret sharing schemes were characterized in terms of
graphical conditions. 
Define the odd neighbourhood of a set $S\subseteq V(\mathsf{G})$ as 
\ben
Odd(S) =  \{v\in V(\mathsf{G}) \mid |N_v\cap D| = 1 \bmod 2 \}\label{eq:oddN}
\een

\begin{proposition}[Authorized sets  for CC protocol, \cite{kashefi09}]\label{prop:authSets}
For the CC classical secret sharing protocols $(G,A)$ of 
\cite{markham08}, the secret can be accessed by a set $S$ if there exists 
$D\subseteq S $  such that 
\ben
D \cup Odd(D) \subseteq S \label{eq:acc0}\\
|D\cap A| =1 \bmod 2\label{eq:acc1}
\een
\end{proposition}

\begin{proposition}[Unauthorized sets for CC protocol, \cite{kashefi09}]\label{prop:unauthSets}
For the CC classical secret sharing protocols of \cite{markham08} on $G$, the secret
cannot be accessed by a set $S$ if there exists a $K\in V(G)\setminus S$ such that
\ben
Odd(K) \cap S =A \cap S\label{eq:noacc}
\een
\end{proposition}

The authors of \cite{kashefi09} proved that these two conditions were sufficient and made the observation  that it was open which graphs satisfy them.
That these conditions are necessary as well was shown in  \cite[Lemma~2]{javelle11}. 

\section{Graph state scheme for general access structures}

\subsection{Classical secrets}

In this section we make a connection between
the CC protocol in the graph state formalism and the standard error correction model.
We  establish a correspondence between all graph state schemes sharing one bit and
$[[n,1]]$ binary quantum codes. This provides an alternative characterization
of the access structure of the CC secret sharing protocols. 
Further, Theorem~\ref{th:ccGF4} also generalizes the results of \cite{ps09}, which only
uses CSS codes derived from self-dual codes.

\begin{theorem}\label{th:ccGF4}
Let $Q$ be an $[[n,1,d]]$ quantum code with stabilizer matrix 
\ben
S= \left[\ba{ccc|ccc} I_r & A_1& A_2 & B & 0 &C \\0 & 0 &0 & D & I_{n-r-1} & E\ea\right],\een
where $\text{diag}(B+CA_2^t)= 0$.
Then the  graph $\ms{G} $ with the adjacency matrix $A_{\ms{G}}$ 
\ben
A_{\ms{G}} = \left[ \ba{ccc}B+CA_2^t & A_1 & A_2\\ A_1^t& 0 & 0 \\A_2^t& 0 &0  \ea\right],\label{eq:adjnGF4}
\een
gives rise to a CC quantum secret sharing protocol with $(\ms{G},A)$, where $A=\supp([\ba{ccc} C^t& E^t &1\ea])$. A generating set for the access structure is given by 
\ben
\Gamma_{\rm{gen}} =\left\{ \supp(g) |  g \mbox{ is an encoded } Z \mbox{ operator. } \right\}
\label{eq:ccGenAccess}
\een
\end{theorem}
\begin{proof}
One choice of logical $X$ and $Z$ operators for $Q$ is given by 
\ben
 \left[\ba{c} \overline{X} \\ \overline{Z} \ea \right] =\left[\ba{ccc|ccc} 0 & E^t & 1& C^t & 0 &0 \\
0 & 0 & 0 & A_2^t& 0 & 1 \ea \right].
\een
Let $\ket{\overline{0}}$ be the state stabilized by $S$ and $\overline{Z}$
and $\ket{\overline{1}} = \overline{X} \ket{\overline{0}}$. 
Then $I^{\otimes r} H^{\otimes n-r}\overline{X}^s \ket{\overline{0}} = Z_A^s\ket{\ms{G}}$.
Therefore, up to  local Clifford unitaries, the basis states of the CC secret sharing scheme induced by
$(\ms{G},A)$ and the basis states of $Q$ are equivalent. Therefore the secret can be recovered if we can
distinguish between the states $\ket{\overline{0}}$ and $\ket{\overline{1}}$.
Consider 
We can rewrite the encoding for the CC protocol in terms of the basis states of $Q$
as follows:
\be
\mc{E}: s\mapsto \overline{X}^s \ket{\overline{0}} =\ket{\psi_s}.
\ee

If $\omega \subseteq \{1,\ldots, n \}$ contains the support of an encoded $\overline{Z}$ operator, then
we can recover the secret because $\overline{Z} \overline{X}^s \ket{\overline{0}} = (-1)^s \overline{X} 
\ket{\overline{0}} $.  Thus the support of any  encoded $Z$ operator gives an authorized set.

If $\omega $ does not contain the support of a
logical $Z$ operator, then it is an unauthorized set.
Let $T$ be such an operator such that $\omega\supseteq \supp(T) \not\supseteq \supp(\overline{Z}M)$, for 
any   $M\in S$. Let $C(S)$ be the centralizer of $S$. 
If $T \not\in C(S)$, then $T$ is detectable, therefore 
$\bra{\psi_s}  T \ket{\psi_s}  =  0$. If $T\in S$, then $\bra{\psi_s}  T \ket{\psi_s}  =  1$.
If $T\in C(S)\setminus S$ and does not contain the 
support of an encoded $Z$ operator, then it must  be an encoded $X$ or $Y$ operator. Since $\ket{\psi_s} =\overline{X}^s\ket{\overline{0}}$, we have $\bra{\psi_s}  T \ket{\psi_s}   = \bra{\overline{0}}\overline{X}^s T \overline{X}^s \ket{\overline{0}} =  0 $, where we used the fact that $T\ket{\overline{0}} =\alpha \ket{\overline{1}}$ for some $\alpha \in \C$. Therefore, by Lemma~\ref{lm:accessConditions}, $\omega$ is unauthorized. This shows that the access structure 
generated by $\Gamma_{\rm{gen}}$ is complete and must coincide with the access structure as
defined by Propositions~\ref{prop:authSets}~and~\ref{prop:unauthSets}.
\end{proof}

\begin{remark}
The requirement that $B+CA_2^t=0$ is not a restriction because any such code can be transformed through
local Clifford unitaries to a code which satisfies this condition. These two codes will lead to the
same access structure. 
\end{remark}

Our theorem gives a succinct characterization of the access structure, we just need to specify the
stabilizer generators and the encoded $Z$ operator. All the authorized sets can then be enumerated
easily. There is no need to further check any other conditions on the sets. But note that our 
characterization does not give the minimal access structure but rather a generating set for the 
access  structure. If we want to obtain the minimal access structure then we only need to look at
those encoded $Z$ operators which are also minimal in the sense they do not properly contain any other 
encoded $Z$ operator within their support.

\begin{corollary}\label{co:ccGF4}
Let $\ms{G}$ be a connected graph with adjacency matrix $A_{\ms{G}}$ as in equation~\eqref{eq:adjnGF4}, where $A_1$, $A_2$ are chosen arbitrarily and $B+CA_2^t$ is symmetric. Let $A$ be an arbitrary
subset of one of the bipartitions of the graph. Then $(\ms{G},A)$ is a CC quantum secret sharing scheme. 
\end{corollary}

Two special cases of are worth highlighting because of their importance. 
\begin{compactenum}[(i)]
\item  $A_1=0$, then it can be
seen that we are covering in effect all possible graphs. Thus every graph leads to a CC quantum 
secret sharing scheme. 
\item  $B+CA_2^t=0$. This corresponds to the situation where $\ms{G}$
is bipartite. These secret sharing schemes are precisely those arising from an $[[n,1]]$ CSS code. 
\end{compactenum}

If $A_2=0$, then the access structure is trivial; the minimal access structure contains a singleton set. 

The framework as developed in \cite{gottesman00}
makes it possible to use mixed states for sharing classical secrets. At the present
it is not clear how to include those schemes in the graph state formalism, since graph states are by definition pure and they lead to pure state quantum secret sharing schemes. Recall that a 
pure state scheme is one in which a pure state is encoded into a pure state. In a mixed state scheme a
pure state could be encoded into a mixed state. Such schemes could be more efficient than the 
pure state schemes. 

\subsection{Quantum secrets}

In this section, we use  the graph state formalism to construct QQ quantum secret sharing  schemes for 
general access structures. Every secret sharing scheme includes a step where the dealer encrypts
the secret before  distributing the shares. In \cite{markham08,keet10}, this was broken down into the following. 
steps. In the first step, the dealer prepares a graph state over the dealer's qubit and the players qubits. Then an ancilla qubit prepared in the secret state is entangled with the dealer's qubit. Then
the ancilla and dealer's qubits are measured in the Bell basis leading to an encoded teleporation onto
the players' qubits. In this paper we simplify these steps by involving only one additional qubit.
We make use of the teleportation scheme to encode into a quantum code using  graph states, see \cite{hein04, grassl11}.

We illustrate this procedure through an example. Consider the graph shown in Fig.~\ref{fig:hamming}. 
Pick any vertex of the graph, say we pick $0$.
The dealer prepares this qubit in the secret state to be shared. Then this qubit is entangled with
the qubits in $N_0$ using controlled-Z gates. Then we measure the dealer's qubit in the $\sigma_x$
basis. If we measure $0$, then the secret has been encoded as desired, otherwise, we need to apply a
correction of the encoded $Z$ on the state. The qubits are then distributed to the players.

\begin{figure}
\begin{tikzpicture}[scale=.75] 
\draw (0,5)[color=black,thick,dashed]--(3,5);
\draw (0,5)[color=black,thick,dashed]--(3,3);
\draw (0,5)[color=black,thick,dashed]--(3,-1);

\draw (0,3)[thick]--(3,3);
\draw (0,3)[thick]--(3,1);
\draw (0,3)[thick]--(3,-1);

\draw[color=red,thick] (0,1)--(3,5);
\draw[color=red,thick] (0,1)--(3,1);
\draw[color=red,thick] (0,1)--(3,-1);

\draw[color=blue,thick] (0,-1)--(3,5);
\draw[color=blue,thick] (0,-1)--(3,3);
\draw[color=blue,thick] (0,-1)--(3,1);

\draw (0,-1) [fill=black] circle (4pt) node [left] {3};
\draw (0,1) [fill=black] circle (4pt) node [left] {2};;
\draw (0,3) [fill=black] circle (4pt) node [left] {1};;
\draw (0,5) [fill=black] circle (4pt) node [left] {0};;

\draw (3,-1)[fill=black]  circle (4pt) node [right] {7};;
\draw (3,1)[fill=black] circle (4pt) node [right] {6};;
\draw (3,3)[fill=black] circle (4pt) node [right] {5};;
\draw (3,5)[fill=black] circle (4pt) node [right] {4};;

\end{tikzpicture}
\caption{(Color online) A general QQ secret sharing scheme from a bipartite graph. All qubits except the dealer's qubit are prepared in the $\ket{+}$ state, while the dealer's qubit (0) is prepared in the secret state.  Then we apply CZ gates along the edges of $\ms{G}$. The dealer's qubit is then measured
in the $\sigma_x$ basis. A correction operator is applied if we measure one. 
 }\label{fig:hamming}
\end{figure}
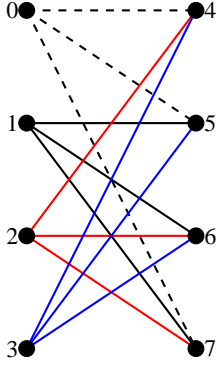

Consider the secret being encoded into $Z_{A}^s\ket{\mathsf{G}\setminus 0}$, where $A=\{4,5,7 \}$.
Then  it can be verified that all the minimal authorized sets given in $\Gamma_{0,\min}$
satisfy both equations~\eqref{eq:acc0}~and~\eqref{eq:acc1}.
\ben
\Gamma_{0,\min} = \left\{\begin{array}{c} \{1,2,7 \}; \{1,3,5\}; \{1,4,6\};  \{2,3,4\}; \\
\{2,5,6\};  \{3,6,7 \}; \{4,5,7\} \end{array}\right\}. \label{eq:hammAccess}
\een

\begin{lemma}\label{lm:bpOdd}
Let $\ms{G}$ be a bipartite graph with adjacency matrix $A_{\ms{G}}$
given by 
\ben
A_\mathsf{G}= \left[\ba{cc} 0 & P\\P^t &0  \ea\right], \mbox{ where } PP^t=I.
\een
Then for any set $D \subseteq B$, where $B$  is in one of the bipartitions,
\begin{compactenum}[(i)]
\item  $Odd(Odd(D))=D$.
\item  $|Odd(D) \cap  N_i| = 0 \bmod 2$ for any $i\in B\setminus D$.
\end{compactenum}
\end{lemma}
\begin{proof}
Let $g\in\F_2^n$ be such that $\supp(g)=D$. Then $Odd(D) = \supp(gA_{\ms{G}})$
and $Odd(Odd(D)) = \supp(gA_{\ms{G}} A_{\ms{G}}) = \supp(g)$.
Let $h\in \F_2^n$ be such that $\supp(h)=N_i$. Let $h$ correspond to the $i$th row in $A_{\ms{G}}$.
Note that $Odd(D)= \supp(A_{\ms{G}}g^t)$. To show that $|Odd(D)\cap N_i|= 0 \bmod 2$, 
it is enough to prove that $h A_\ms{G} g^t=0 $. This is zero because $h$ is orthogonal to all but the
$i$the column of $A_\ms{G}$. 
\end{proof}

\begin{theorem}\label{th:qqCSS}
Let $\mathsf{G}$ be a bipartite graph whose adjacency matrix $A_{\mathsf{G}}$ is given by 
\ben
A_\mathsf{G}= \left[\ba{cc} 0 & P\\P^t &0  \ea\right], \mbox{ where } PP^t=I.
\een
Then for every vertex $i$ we can define a perfect QQ  quantum secret sharing scheme from $\mathsf{G}$.
The encoding for the quantum secret sharing scheme is given by 
\ben
\mc{E} : a\ket{0}+b\ket{1}\mapsto a\ket{\mathsf{G}\setminus i} +b Z_{N_i}^s\ket{\mathsf{G}\setminus i}.
\label{eq:qqEncrypt}
\een
A generating set for the access structure $\Gamma_i$ is given by the following
\ben
\Gamma_{i,\rm{gen}} &= & \left\{D \cup Odd(D)\setminus i \bigg|\ba{c} D \subseteq V_r \\ |D\cap N_i| = 1 \bmod 2\ea\right\},\label{eq:qqGenAccess}
\een
where $V_r$ is the bipartition of vertices of $\ms{G}$ that does not contain $i$.
The encryption and recovery of the secret are as shown in Fig.~\ref{fig:encrypyQQ}~and~\ref{fig:recoverQQ} respectively. 
\end{theorem}
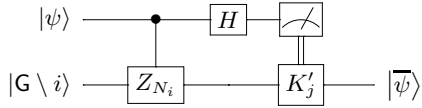
\begin{figure}[htb]
\[
\Qcircuit @C=.7em @R=.4em @! {
  \lstick{\ket{\psi}}  & \ctrl{1} &
    \gate{H} & \meter & \\
  \lstick{\ket{\ms{G}\setminus i}}  & \gate{Z_{N_i}} &
    \qw & \gate{K_j'} \cwx &  
    \rstick{\ket{\overline{\psi}}} \qw
}
\]
\caption{Encrypting the secret state $\ket{\psi}$ for QQ secret sharing using teleportation. The operator $K_j'=X_j\prod_{k\in N_j\setminus i}Z_k$ is such that
$j\in N_i$. It is applied only if the measurement outcome is 1.}\label{fig:encrypyQQ}
\end{figure}

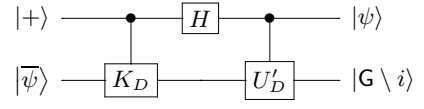
\begin{figure}[htb]
\[
\Qcircuit @C=.7em @R=.4em @! {
  \lstick{\ket{+}}  & \ctrl{1} &
    \gate{H} & \ctrl{1} & \rstick{\ket{\psi}} \qw\\
  \lstick{\ket{\overline{\psi}}}  & \gate{K_D} &
    \qw & \gate{U_D'}  &  
    \rstick{\ket{\ms{G}\setminus i}} \qw
}
\]
\caption{Reconstructing the secret state $\ket{\psi}$ for QQ secret sharing given an authorized set 
$D$ as in equation~\eqref{eq:qqGenAccess}. The operator $K_D=\prod_{j\in D }K_j = \prod_{j\in D}X_j \prod _{k\in Odd(D)}Z_k$ and $U_D' =  
\prod_{j\in D}Z_j \prod _{k\in Odd(D)\setminus i}X_k$.}\label{fig:recoverQQ}
\end{figure}

\begin{proof}
We shall prove this theorem in parts. For convenience, we shall ignore the normalization factors for
quantum states.
\begin{compactenum}[(i)]
\item Encryption of the secret:
Assume that the secret to be encoded is $\ket{\psi} = a\ket{0} +b\ket{1}$. Then 
it can be easily verified that in Fig.~\ref{fig:encrypyQQ}, the state $\ket{\psi}\ket{\ms{G}
\setminus i }$ is  transformed to the following state prior measurement (up to normalization): 
\be 
\ket{0}(a\ket{\ms{G}\setminus i}+bZ_{N_i}\ket{\ms{G}\setminus i})+
\ket{1}\left(a\ket{\ms{G}\setminus i} - bZ_{N_i}\ket{\ms{G}\setminus i}\right)
\ee
If we measure zero, then we get the desired state but if we measure one, then we have to apply the correction operator $K_j'= X_j\prod_{k\in N_i\setminus i }Z_k$ for any $j\in N_i$. This operator
anti commutes with the operator $Z_{N_i}$,  but stabilizes the state $\ket{\ms{G}\setminus i}$,
therefore it acts as a correction operator to give the state in equation~\eqref{eq:qqEncrypt}.
 
 \item Recovery: 
Before we show that $D\cup Odd(D)\setminus i $ is authorized, we need the following result. 
Note that $K_D=\prod_{j\in D} K_j $, therefore $ \supp(K_D)= D\cup Odd(D)$ and $i\in \supp(K_D)$. Because of the fact
$PP^t=I$, the element $U_D=X_iU_D'=\prod_{j\in D} Z_j \prod_{k\in Odd(D)}X_k \in S(\ket{\ms{G}})$. 
Let us  write $\ket{\ms{G}}$
as 
\be
\ket{\ms{G}} &=& \ket{0}\ket{\ms{G}\setminus i}+\ket{1}Z_{N_i}\ket{\ms{G}\setminus i},
\ee
Then  it follows that 
\be
U_D\ket{\ms{G}} &=& \ket{1}U_D'\ket{\ms{G}\setminus i}+\ket{0}U_D'Z_{N_i}\ket{\ms{G}\setminus i}.
\ee 
Hence, 
 $U_D'Z_{N_i}$ stabilizes $\ket{\ms{G}\setminus i }$. 
 
With respect to the recovery observe that the set $D\cup Odd(D)$  as in Eq.~\eqref{eq:qqGenAccess}
satisfies the requirements of Proposition~\ref{prop:authSets}, therefore if we trace through the
circuit given in Fig.~\ref{fig:recoverQQ}, the state transforms as follows:
\be
\ket{+}\ket{{\overline\psi}} & = & (\ket{0}+\ket{1})
(a\ket{\ms{G}\setminus i}+bZ_{N_i}\ket{\ms{G}\setminus i})\\
&\stackrel{c-K_D}{\longrightarrow}& \ket{0}(a\ket{\ms{G}\setminus i}+bZ_{N_i}\ket{\ms{G}\setminus i})+\\
& &
\ket{1}(aK_D\ket{\ms{G}\setminus i}+bK_DZ_{N_i}\ket{\ms{G}\setminus i})\\
&\stackrel{H}{\longrightarrow}& a\ket{0}\ket{\ms{G}\setminus i} + 
b\ket{1}Z_{N_i}\ket{\ms{G}\setminus i}\\
&\stackrel{c-U_D'}{\longrightarrow}& a\ket{0}\ket{\ms{G}\setminus i} + 
b\ket{1}\ket{\ms{G}\setminus i}\\
& =& (a\ket{0}+b\ket{1})\ket{\ms{G}\setminus i}
\ee
where we used the fact that $U_D' Z_{N_i}$ stabilizes $\ket{\ms{G}\setminus i}$.
Thus $D\cup Odd(D)\setminus i$ is able to reconstruct the quantum secret $\ket{\psi}$.
The no-cloning theorem now implies that the complement of this set is unauthorized. 

\item Completeness of $\Gamma_{i,\rm{gen}}$:
Now we show that the access structure as defined in Eq.~\eqref{eq:qqGenAccess} is 
complete in the sense that every authorized set contains some element of $\Gamma_{i,\rm{gen}}$.

Assume that there exists some set $A$ which is authorized
but not generated by $\Gamma_{i,\rm{gen}}$. The complement of this set is unauthorized. Let this be denoted as
$B$. Let $A_i =A\cap N_i$ and $B_i = B\cap N_i$, then $B_i = N_i\setminus A_i$.
Since $B$ is unauthorized, it must have $|B_i|=0\bmod 2$ or $Odd(B_i) \not\subseteq B$.
(a) Suppose that it is the case that $|B_i|=0\bmod 2$. Since $PP^t=I$ , $|N_i| = 1\bmod 2 $ and
$|B_i|=0\bmod 2$, it follows that $|A_i|=1\bmod 2$. If $Odd(A_i)\subseteq A$, then it is 
 generated by $\Gamma_{i,\rm{gen}}$, therefore, it must be the case that $Odd(A_i) \not\subseteq A $.
By Lemma~\ref{lm:bpOdd},  $Odd(Odd(i))= Odd(N_i)=i$ it follows that $Odd(B_i)  = Odd(A_i)\setminus \{i\} $. Consider the set 
$ C= Odd(B_i) \setminus A$. This has support only in the unauthorized set $B$. Then $Odd(C)\cup B_i$
has an odd neighbourhood which is given by $ Odd(B_i)\setminus C$ and it lies entirely in 
$A$. Further, the set $A'= A\setminus Odd(C)$ has an odd neighbourhood $Odd(B_i)
\setminus C \cup {i}$. Thus the new set $A' \subseteq A$ is an element of $\Gamma_{i,\rm{gen}}$
contrary to assumption that $A$ is not  generated by $\Gamma_{i,\rm{gen}}$. 

(b) If on the other hand, $|B_i|=1\bmod 2$, then by exactly the same argument but reversing the roles of
$A$ and $B$, we see that $B$ is generated by an element of $\Gamma_{i,\rm{gen}}$ and is therefore
authorized. This would violate the no-cloning theorem as $A$ and $B$ are disjoint. Thus it is not 
possible for any authorized set to exist outside the access structure generated by $\Gamma_{i,\rm{gen}}$.

\item Perfectness of $\Gamma_{i,\rm{gen}}$: We also need to show that this scheme is perfect, namely,
there are no unauthorized sets which although unable to reconstruct the secret are still able to extract some information about the secret. Because it is a pure state scheme by construction, it is sufficient
to show that the complement of an unauthorized set is authorized. This can be shown using 
almost the same argument as in (iii) but this time assuming that $A$ and its complement $B$ are both unauthorized.
In this case arguing as (iii a), we would conclude that $A \in \Gamma_i$ contrary to the assumption
that $A$ is unauthorized or else arguing as in (iii b) we would conclude that $B$ is authorized.
This ensures that the QQ scheme constructed is perfect. 
\end{compactenum}
\end{proof}

The access structure realized by this scheme is same as the access structure realized by the 
quantum secret sharing scheme using the approach of quantum error correcting codes as the following
result shows.

\begin{corollary}
Let $Q$ be an $[[n,0]]$ CSS code, with the stabilizer matrix, 
\ben
S= \left[\ba{cc|cc}I & P & 0 & 0 \\0 & 0 &I &P   \ea\right],
\een 
where $PP^t=I$.
Then the $[[n-1,1]]$ quantum code obtained by puncturing the $i$th qubit realizes the QQ secret sharing protocol 
of Theorem~\ref{th:qqCSS}.
\end{corollary} 
\begin{proof}
It suffices to show that the quantum states in Eq.~\eqref{eq:qqEncrypt} form a basis for the quantum
code obtained by puncturing the $i$th qubit. 
Without loss of generality we can assume that we puncture the 0th qubit. 
Let $P= \left[ \ba{c}g\\ Q\ea \right]$, where $(0|g) \in \F_2^n$
and $\supp(0|g)= N_i$. Note that $g\neq 0$ because of the requirement $PP^t=I$.
Then puncturing the $i$th qubit results in an $[[n-1,1]]$
quantum code. The stabilizer matrix for this code is 
\be
\left[\ba{cc|cc}I & Q & 0 & 0 \\ 0 & 0 & I & Q \ea\right],
\ee
while the encoded operators are given by 
\be
\left[\ba{c} \overline{X}\\ \overline{Z} \ea\right]= \left[\ba{cc|cc} 0 & g & 0 & 0 \\0 & 0 & 0 & g 
\ea\right].
\ee
Consider the state stabilized by $S$ and $\overline{Z}$. Its stabilizer matrix is 
\be
\left[\ba{cc|cc}I& Q & 0&0 \\0 & 0 & 0 & g\\0 & 0 & I & Q \ea \right].
\ee
This matrix is equivalent to the following, through row transformations,
\be
\left[\ba{cc|cc}I& Q & 0&0 \\0 & 0 & Q^t & I \ea \right],
\ee
which is precisely the stabilizer of the state $I^{\otimes n/2-1} H^{\otimes n/2}\ket{\ms{G}\setminus i }$. The state stabilized by $S$ and $\overline{X}$ on the other hand is 
$I^{\otimes n/2-1} H^{\otimes n/2}Z_{N_i}\ket{\ms{G}\setminus i }$.
\end{proof}

\begin{remark}
In this paper we have only considered secret sharing schemes where the share distributed to each
party is of the same dimension as the dimension of the secret. Such schemes are said to be ideal. 
The graph state framework has not been used to study schemes which are not ideal. Such a need
arises because there are some schemes that are not ideal.
\end{remark}

Our results make it possible to answer some questions related to the graph state formalism very
easily as exemplified by the following theorem.

\begin{theorem}
There do not exist any graph state QQ secret sharing protocols for $((k,2k-1))$ if $k\geq 4$.
\end{theorem}
\begin{proof}
In \cite{rietjens05}, it was shown that every $((k,2k-1))$  quantum threshold secret sharing scheme is 
an $[[2t-1,1,t]]$ quantum MDS code. In \cite{calderbank98}, it was shown that there do not exist any
$[[n,1]]$ binary quantum MDS codes of length greater than $5$. It follows therefore, there are no
(pure state) QQ quantum threshold schemes of length $2k-1$ greater than 5, equivalently $k\geq 4$. 
\end{proof}

The existence of quantum threshold schemes  was studied at great length in \cite{javelle11}. 
Through the connection to quantum codes we are able to shed light on this issue,
 immediately improving upon the lower bound in \cite[Corollary~4]{javelle11}.

\section{Conclusion}
We showed how to construct quantum secret sharing schemes with arbitrary access structures
using the graph state formalism. These results also elucidate the connection between graph
state framework and quantum secret sharing schemes based on quantum codes \cite{cleve99,gottesman00}.

\section*{Acknowledgment}
I would like to thank Ben Fortescue for helpful discussions and Ken Brown for supporting this
research through a grant from IARPA.


\begin{thebibliography}{10}

\bibitem{calderbank98}
A.R. Calderbank, E.M. Rains, P.W. Shor, and N.J.A. Sloane.
\newblock Quantum error correction via codes over {GF}(4).
\newblock {\em IEEE Trans. Inform. Theory}, 44:1369--1387, 1998.

\bibitem{cleve99}
R.~Cleve, D.~Gottesman, and H.-K. Lo.
\newblock How to share a quantum secret.
\newblock {\em Phys. Rev. Lett.}, 83(3):648--651, 1999.

\bibitem{gottesman97}
D.~Gottesman.
\newblock Stabilizer codes and quantum error correction.
\newblock {C}altech Ph. D. Thesis, eprint: quant-ph/9705052, 1997.

\bibitem{gottesman00}
D.~Gottesman.
\newblock Theory of quantum secret sharing.
\newblock {\em Phys. Rev. A}, 61(042311), 2000.

\bibitem{grassl11}
M.~Grassl.
\newblock Variations on encoding circuits for stabilizer quantum codes.
\newblock In {\em Proceedings Third International Workshop Coding and
  Cryptology, Lecture Notes in Computer Science}, pages 142--158, 2011.

\bibitem{gravier11}
S.~Gravier, J.~Javelle, M.~Mhalla, and S.~Perdrix.
\newblock On weak odd domination and graph-based quantum secret sharing, 2011.
\newblock eprint:arXiv:1112.2495.

\bibitem{hein04}
M.~Hein, J.~Eisert, and H.~J. Briegel1.
\newblock Multiparty entanglement in graph states.
\newblock {\em Phys. Rev. A}, 69(062311), 2004.

\bibitem{hillery99}
M.~Hillery, V.~Buzek, and A.~Berthaume.
\newblock Quantum secret sharing.
\newblock {\em Phys. Rev. A}, 59(3):1829--1834, 1999.

\bibitem{javelle11}
J.~Javelle, M.~Mhalla, and S.~Perdrix.
\newblock New protocols and lower bound for quantum secret sharing with graph
  states.
\newblock eprint:arXiv:1109.1487, 2011.

\bibitem{kashefi09}
E.~Kashefi, D.~Markham, M.~Mhalla, and S.~Perdrix.
\newblock Information flow in secret sharing protocols, 2009.
\newblock eprint:arXiv:0909.4479.

\bibitem{keet10}
A.~Keet, B.~Fortescue, D.~Markham, and B.~C. Sanders.
\newblock Quantum secret sharing with qudit graph states.
\newblock eprint:arXiv:1004.4619, 2010.

\bibitem{markham08}
D.~Markham and B.~Sanders.
\newblock Graph states for quantum secret sharing.
\newblock {\em Phys. Rev. A}, 78(042309), 2008.

\bibitem{rietjens05}
K.~Rietjens, B.~Schoenmakers, and P.~Tuyls.
\newblock Quantum information theoretical analysis of various constructions for
  quantum secret sharing.
\newblock In {\em Proc. 2005 IEEE Intl. Symposium on Information Theory,
  Adelaide, Australia}, pages 1598--1602, 2005.

\bibitem{ps09}
P.~Sarvepalli and A.~Klappenecker.
\newblock Sharing classical secrets with {C}alderbank-{S}hor-{S}teane codes.
\newblock {\em Phys. Rev. A}, 80(022321), 2009.

\bibitem{smith00}
A.~Smith.
\newblock Quantum secret sharing for general access structures.
\newblock eprint: arXiv:quant-ph/0001087, 2000.

\end{thebibliography}

\end{document}